\documentclass[a4paper,11pt]{article}

\usepackage{fullpage}

\usepackage{color}
\usepackage{graphicx}
\usepackage{listings}

\usepackage{amsmath}
\usepackage{amsthm}

\usepackage{latexsym}
\usepackage{url}
\usepackage{xspace}

\usepackage{booktabs}
\usepackage{cite}

\usepackage{algorithm}
\usepackage{algpseudocode}

\usepackage{listings}
\usepackage{hyperref}

\newcommand{\ceil}[1]{\lceil #1\rceil}

\newcommand{\mpibarrier}{\texttt{MPI\_\-Barrier}\xspace}
\newcommand{\mpialltoall}{\texttt{MPI\_\-Alltoall}\xspace}
\newcommand{\mpialltoallv}{\texttt{MPI\_\-Alltoallv}\xspace}
\newcommand{\mpialltoallw}{\texttt{MPI\_\-Alltoallw}\xspace}
\newcommand{\mpidimscreate}{\texttt{MPI\_\-Dims\_\-create}\xspace}
\newcommand{\mpicommsplit}{\texttt{MPI\_\-Comm\_\-split}\xspace}

\newcommand{\mpicontig}{\texttt{MPI\_\-Type\_\-contiguous}\xspace}
\newcommand{\mpiresized}{\texttt{MPI\_\-Type\_\-create\_\-resized}\xspace}

\newcommand{\mpiint}{\texttt{MPI\_\-INT}\xspace}

\newcommand{\hydraopenmpi}{OpenMPI~4.1.6\xspace}
\newcommand{\hydraintelmpi}{IntelMPI~2021.13\xspace}
\newcommand{\hydrampich}{\texttt{mpich 4.1.2}\xspace}
\newcommand{\gcc}{\texttt{gcc 10.2.1}\xspace}

\newtheorem{theorem}{Theorem}

\title{Effective MPI: User-defined Datatypes and Cartesian Communicators for Zero-copy All-to-all Communication in Multidimensional Tori}

\author{Jesper Larsson Tr\"aff\\
  TU Wien\\
  Faculty of Informatics\\
  Institute of Computer Engineering, Research Group Parallel Computing 191-4\\
  Treitlstrasse 3, 5th Floor, 1040 Vienna, Austria}
\date{May 2026}

\begin{document}
\maketitle

\begin{abstract}
We present and show how to implement a non-trivial all-to-all
communication algorithm for arbitrary $d$-dimensional tori effectively
in MPI. Given a factorization of the number of processes $p$ into $d$
factors that can be mapped onto a $d$-dimensional torus, we first
utilize a Cartesian communicator to split a given $p$-process MPI
communicator into, for each MPI process, $d$ smaller communicators
spanning each of the dimensions of the torus to which the process
belongs, and cache these communicators in order to avoid expensive
splitting at each all-to-all operation. The all-to-all operation
itself is decomposed into a sequence of $d$ \mpialltoall operations on
the dimension-wise communicators. The non-trivial data rearrangement
before and after each \mpialltoall call is implicit only and effected
by MPI derived datatypes. This makes the implementation of the
algorithm formally \emph{zero-copy}, meaning that no explicit
process-local reordering of data blocks ever has to be performed. In
order to achieve this, the algorithm employs a double-buffering scheme
with modest temporary buffer requirements. By choosing the
factorization of $p$ and selecting appropriate implementations for the
component \mpialltoall operations, the presented implementation gives
ample opportunities for algorithm tuning and adaptation to the
particular high-performance system. A few, select experimental results
show competitive performance with native \mpialltoall implementations
and illustrate problems that common \mpialltoall implementations may
have.
\end{abstract}

\section{Introduction}

The all-to-all personalized communication problem, in which each
process in a set of processes has an individual block of data (of some
given number of elements) to each of the other processes, arises
naturally in many important applications, but is also the most
communication intensive and expensive collective data exchange
problem. The operation should therefore be used judiciously, while at
the same time efficient and well-implemented algorithms are needed.
All-to-all communication is standardized in MPI as the regular
\mpialltoall collective and the irregular \mpialltoallv and
\mpialltoallw collectives~\cite{MPI-4.1}.

We present a paradigm for potentially efficient implementations of the
all-to-all operation in terms of smaller component all-to-all
operations, which shows how such a collective operation can be
expressed in terms of more specialized uses of the same collective
operation. Such a decomposition is possible with MPI functionality and
highlights how MPI can be used effectively to efficiently
implement non-trivial algorithms. Such implementations of MPI
functionality in terms of other MPI functionality naturally poses
upper bounds on the concrete performance of, in the case here, the
\mpialltoall collective in any given MPI library as so-called
\emph{self-consistent performance guidelines}~\cite{Traff10:selfcons}. Such
performance guidelines can be automatically checked and strengthen
ones confidence in the performance and quality of a given MPI
library~\cite{Traff16:autoguide}. Essentially, a performance guideline
states that for a given collective operation, there is no immediate
better implementation in terms of MPI itself.

We can view the all-to-all algorithm that we present here in three different
ways.
\begin{enumerate}
\item
  As a dedicated algorithm for a torus communication network of
  any number of dimensions, in which communication is restricted to
  be between processes (or nodes of processes) that are neighbors along
  one of the dimensions in the torus. With $d$ dimensions, the algorithm
  takes $d$ communication rounds.
\item
  As an algorithm for the all-to-all problem for any given or desired
  factorization of the number of processes $p$. Viewed in this way,
  the algorithm can be adopted to the given communication network and
  partition by choice of factorization and choice of implementations
  for the all-to-all problems to be solved on the factors. With $d$
  factors, the algorithm takes $d$ communication rounds. For a power
  of two number or processes, $p=2^d$, this is $\log_2 p$, and the
  resulting algorithm can be seen as a natural generalization of
  well-known hypercube algorithms.
\item
  As message-combining performance guideline implementations of the
  all-to-all operation against which the MPI library native
  \mpialltoall function can be judged by comparative
  benchmarking~\cite{Traff16:autoguide}.
\end{enumerate}

All three viewpoints entail combining of blocks of elements from and
to different processes. In each round different combinations of blocks
are needed. We show how such composite blocks can be put together
implicitly, at the corresponding component \mpialltoall call, by
expressing the sequences of blocks as derived datatypes. The
communicators for the component \mpialltoall operations are created in
advance and cached so that the expensive creation of the needed
subcommunicators can be amortized over many all-to-all operations.
The paper has the didactic purpose of showing the effectiveness of MPI
mechanisms (Cartesian communicators, communicator splitting, attribute
caching and derived datatypes) for the implementation of non-toy
examples.

\section{Related Work}

Standard all-to-all algorithms for two- and three-dimensional meshes
and hypercubes are explained in~\cite{GramaKarypisKumarGupta03}, see
also~\cite{SaadSchultz89:hypercube}
and~\cite{BertsekasOzverenStamoulisTsengTsitsiklis91} for full-ported
hypercubes. Message combining all-to-all algorithms like
Bruck~\cite{Bruck97}, and for more general circulant
graphs~\cite{Traff25:reducescatter}, are well-known and used in many
different forms, also in MPI libraries. Message-combining can reduce
the number of communication rounds $d$ down to the optimal
$d=\ceil{\log_2 p}$, but at an inherent cost of (re)sending roughly
half the data blocks in each round~\cite{Bruck97}. An implementation
of the Bruck algorithm using MPI derived datatypes to reorganize and
pipeline data blocks was given in~\cite{Traff14:bruck}.

Multidimensional torus all-to-all algorithms with one-ported
communication were presented and implemented
in~\cite{TsengGupta96,TsengNiSheu99,SuhShin01}, and elsewhere, and use
links along only one dimension at a time.
The simple $d=2$ algorithm for the BlueGene/L system
in~\cite{YazakiTakaueAjimaShimizuIshihata12} attempts to use multiple
links simultaneously.  Two-level decompositions of the \mpialltoall
operation, also relying on MPI derived datatypes were given
in~\cite{Traff14:alltoall,Traff20:mpidecomp}; in the latter the
decomposition in particular follows the compute node structure of the
given cluster and performs the all-to-all communication as concurrent
intra-node all-to-all operations followed by concurrent inter-node
all-to-all operations (or the other way round). The algorithm of this
paper extends this implementation to arbitrary number of dimensions
and deeper processor hierarchies. A somewhat similar algorithm, for
two levels only, viewing the MPI all-to-all problem as a matrix
transposition problem, was given in~\cite{ChochiaSoltHursey22}.

Our implementation precomputes and uses information stored with
communicators to decide on the details of the algorithms to be
executed. This viewpoint was discussed in more detail specifically for
MPI in~\cite{Traff21:orthogonality}. Some problems related to
\mpidimscreate that also appear in this paper were discussed and
solved in~\cite{Traff15:dimscreate}, see
also~\cite{NiethammerRabenseifner19}.


\section{A $d$-dimensional Torus All-To-All Algorithm}

Assume that a $p$-process(or) communication structure (network,
system, artifact), say a $d$-dimensional torus, is given and can be
described as a Cartesian graph product $M^D_d=G_0\Box\allowbreak G_1\Box \ldots
\Box\allowbreak G_{d-1}$ with graph factors $G_i=(V_i,E_i)$ and
$p=\Pi_{i=0}^{d-1}D[i]=|V_0|\times |V_1|\times\ldots |V_{d-1}|$. When
$M^D_d$ is a torus with dimension orders $D[i],i=0,1,\ldots,d-1$, the
factors $G_i$ are ring communication structures, when $M^D_d$ is a
mesh, the factors are linear processor arrays; but the factors could
be (isomorphic to) other graphs (like fully connected networks) as
fitting to the actual hardware communication capabilities of the given
system, as long as they, for the algorithm to be described, admit
all-to-all communication.

Our $d$-dimensional torus all-to-all algorithm decomposes the
all-to-all operation into concurrent all-to-all operations on smaller
communication structures along the dimensions of the torus, each
isomorphic to a factor $G_i$ of the graph factorization.  On dimension
$i,0\leq i<d$, a smaller all-to-all operation on $D[i]$ processors
will exchange $p$ blocks of elements (the blocks to the processor
itself do not have to be communicated), and therefore the communication
volume complexity will be proportional to both $p$ and $d$ roughly as
$pd$ (with a constant factor smaller than one). The algorithm works
for any number of torus dimensions and for any factorization of $p$
into $d$ factors (in any chosen order). In particular, if $p=2^d$ is a
power of two, the resulting algorithm will be equivalent to well-known
hypercube algorithms for all-to-all communication.

\begin{algorithm}
  \caption{The generic all-to-all algorithm for processor $r\in M^D_d$
    in a $d$-dimensional Cartesian product $M^D_d=G_0\Box G_1\Box
    \ldots \Box G_{d-1}$ with $G_k=(V_k,E_k)$. Processor $r$ has $p$
    input blocks in $S[p]$, one for each processor, and receive a
    block from each processor in $R[p]$.}
  \label{alg:dgenericalltoall}
  \begin{algorithmic}
    \Procedure{AllToAll}{$S[p],R[p],r\in M^D_d$}
    \State $O\gets \Call{rank-to-vector}{r,d,D}$ \Comment Origin vector
    \State $b\gets 0$ \Comment Double buffer switch
    \State $T[b]\gets S$
    \For{$k=0,1,\ldots,d-1$} \Comment Dimension-wise all-to-all
    \State\Comment Define stride function $\sigma(i) = \Pi_{k'=0}^{i-1}D[k']$
    \State $S'_{[\sigma(k)][\sigma(k+1)]\ldots [\sigma(d-1)]}[D[k]][D[k+1]]\ldots [D[d-1]]\equiv T[b]$
    \State $R'_{[\sigma(k)][\sigma(k+1)]\ldots [\sigma(d-1)]}[D[k]][D[k+1]]\ldots [D[d-1]]\equiv T[1-b]$
    \State $G\equiv M^D_d[\{O[0]\}\times\cdots\times\{O[k-1]\}\times V_k\times\{O[k+1]\}\times \cdots \times \{O[d-1]\}]$
    \State $\Call{AllToAll}{S',R',G}$ \Comment All-to-all on subgraph $G$
    \State $b\gets 1-b$
    \EndFor
    \State $R\gets T[b]$
    \EndProcedure
  \end{algorithmic}
\end{algorithm}

The all-to-all algorithm is shown as
Algorithm~\ref{alg:dgenericalltoall} and relies heavily on a specific
notation for indexing in arrays.  The input and output arrays $S$ and
$R$ are declared as one-dimensional arrays of $p$ blocks (of
elements) in linear index order such that $S[i]$ denotes the $i$th
block of elements to be sent to processor $i$ and $R[i]$ the $i$th
block in $R$ which is to be received from processor $i$. The bracket
$[i]$ can be seen as a mapping which gives the index or address of the
$i$th element of the array to which it is applied. Multiple bracket
are used, C style~\cite{KernighanRitchie88}, for multidimensional
arrays. A declared array $T[D[d-1]]\ldots [D[0]]$ denotes a
$d$-dimensional array of order $p=D[0]\times D[1]\times\ldots\times
D[d-1]$.  Lookup with $d$ indices $T[i_{d-1}][i_{d-2}]\ldots[i_0]$
addresses element
$i=i_{d-1}\Pi_{k=0}^{d-2}D[k]+i_{d-2}\Pi_{k=0}^{d-3}D[k]+\ldots +i_0$,
corresponding to a row-major organization of the elements of $T$. Each
of the factors $\Pi_{k=0}^{d-2}D[k]$ etc.\ is the \emph{stride} used
for the given dimension (here $d-1$). Explicit strides can be used to
express other traversals of the array. In our notation
$T_{[s_{d-1}][s_{d-2}]\ldots [s_0]}[i_{d-1}][i_{d-2}]\ldots [i_0]$
will denote the element $i=i_{d-1}s_{d-1}+i_{d-2}s_{d-2}+\ldots
+i_0s_0$ with explicit strides $s_i$.  For instance, taking
$s_i=\Pi_{j=i}^{d-2}D[j]$, the $d$-dimensional array is handled in
column-major order. See the examples at the end of this section.

The all-to-all algorithm takes $d$ successive all-to-all communication
rounds, in each of which multiple, concurrent all-to-all operations
are performed.  The processors use two intermediate buffers $T[0]$ and
$T[1]$ that are divided into $p$ blocks in the same way as the result
buffer $R[p]$; but indexed differently by each of the all-to-all
operations. Let processor $r\in M^D_d$ have coordinates
$(O[0],O[1],\ldots,O[d-1])$. In communication round $k,
k=0,1,\ldots,d-1$, processor $r$ shall do an all-to-all exchange with
the processors
\begin{displaymath}
  \{O[0]\}\times\cdots\times\{O[k-1]\}\times V_k\times\{O[k+1]\}\times \cdots \times \{O[d-1]\}
\end{displaymath}
which are the processors of the communication structures $G$ that are
used for the all-to-all calls. For each processor, $G$ is the subgraph
of $M^D_d$ restricted to these processors.  The all-to-all operation on
the corresponding communication structures will ensure that processor
$r$ has received all blocks from the processors in its $G$ to itself
\emph{and} to all processors with which $r$ will do all-to-all
operations in the following rounds. All these blocks are combined into
a single composite block for processor $r$.  The communication
structures for round $k$ each obviously have size $D[k]$. The number
of blocks that must be sent and received (the blocks that a processor
communicates with itself not included) in round $k$ is
$(D[k]-1)\frac{p}{D[k]}$, since the number of blocks per processor is
$p/D[k]$. In the first round with $k=0$, processor $r$ has one block
to each of the other processors with which it communicates and to the
processors with which these will communicate in later rounds. After
the round, processor $r$ has received $D[0]$ blocks to itself (from
$D[0]-1$ other processors and itself), and also all the segments of
blocks to the processors with which $r$ will later communicate now
consist of $D[0]$ blocks. In each round, the segments of blocks that
are exchanged in the all-to-all communication get larger by a factor
of $D[k]$ blocks, such that after the last round, each processor will
have received $\Pi_{k=0}^{d-1}D[k]=p$ blocks, including the block from
itself. After each round, the blocks to processor $r$ and the
processors with which $r$ will exchange blocks in the later rounds
will form consecutive segments of the $T[0]$ and $T[1]$ buffers with
each such segment consisting of $\Pi_{i=0}^{k-1}D[i]$ blocks.  The
segments will be put together by traversing the buffers as
$(d-1-k)$-dimensional arrays in column-major order. For round $k$,
this is expressed as
\begin{displaymath}
  \underbrace{S'_{[\sigma(k)][\sigma(k+1)]\ldots [\sigma(d-1)]}[D[k]]}_{\textrm{Composite blocks}}\underbrace{[D[k+1]]\ldots [D[d-1]]}_{\textrm{Segment indices}}\underbrace{\ldots}_{\textrm{Consecutive $\Pi_{i=0}^{k-q}D[i]$ block segment}}
\end{displaymath}
where $\sigma(k)=\Pi_{i=0}^{k-1} D[i]$ is the stride to be used for the
the corresponding index. The first index (from the left) is the number
of composite blocks and therefore equal to the size $D[k]$ (number of
processors) of $G$. The number of indexed segments is $\Pi_{i=k}^{d-1}D[i]$ and
the size of the segments $\Pi_{i=0}^{k-1}D[i] = \sigma(k)$. The second index
enumerates the segments for the next dimension, and are therefore strided
by $\sigma(k+1)$ blocks; and so on for the remaining indices up to $d-1-k$.
We note that the stride for the first (leftmost) index equals the
size of the consecutive segments for round $k$.

\begin{theorem}
  \label{thm:genericalltoall}
  On a $d$-dimensional torus $M^D_d$,
  Algorithm~\ref{alg:dgenericalltoall} correctly solves the all-to-all
  problem in $d$ successive rounds of concurrent all-to-all
  operations. There are $p/D[k]$ concurrent all-to-all operations in
  round $k,k=0,1,\ldots,d-1$, each of which involves $D[k]$
  processors.  The number of elements in the blocks that are sent and
  received in round $k$ is likewise $p/D[k]$. Per processor $D[k]-1$
  such blocks are sent and received, and one block is copied from send
  to receive buffer. The total number of blocks sent and received per
  processor is
  \begin{align*}
    \sum_{k=0}^{d-1}\frac{D[k]-1}{D[k]}p &=dp-\sum_{k=0}^{d-1}\frac{p}{D[k]}
    \quad .
  \end{align*}
\end{theorem}
\begin{proof}
  The complexity in number of all-to-all operations and number and sizes
  of blocks follows immediately from the description of the
  $d$ rounds of the algorithm.

  For correctness, Algorithm~\ref{alg:dgenericalltoall} maintains the
  following invariant: Before round $k$, processor
  $r=(O[0],\ldots,O[d-1])$ has collected a block to each of the
  processors in
  \begin{displaymath}
    V_{d-1}\times\ldots\times V_k\times\{O[k-1]\}\times\ldots\times\{O[0]\}
  \end{displaymath}
  from each of the processors in
  \begin{displaymath}
    \{O[d-1]\}\times\ldots\times\{O[k]\}\times V_{k-1}\times\ldots\times V_0
\quad .
  \end{displaymath}
  The invariant holds before the first iteration $k=0$, since each processor
  has in the input buffer $S[p]$ a block to each of the other processors.
  The invariant implies that after the last iteration, processor $r$ has
  received a block from each of the other processors.

  Assume that the invariant holds before iteration $k$. By the all-to-all
  operation, each processor sends the blocks that it has for all
  processors in $\{O[0]\}\times\ldots\{O[k-1]\}\times
  V_k\times\{O[k+1]\}\times\ldots\times\{O[d-1]\}$ and receives the
  blocks that the other processors have.
\end{proof}

Except for the copy operations before the first iteration of $S$ into
$T[0]$ and after the last iteration of $T[b]$ into $R$, there are no
explicit copy operations. All exchanges and reordering of elements is
handled implicitly by the traversal of the temporary buffers (as
$(d-1-k)$-dimensional arrays) and the all-to-all operations. The first
and last copies can easily be avoided in a real implementation by
using $R$ as one of the intermediate buffers and making sure that this
will be used as receive buffer in the last all-to-all call, and by taking
the block directly from the send buffer $S$ in the first
iteration. Doing this makes the implementation formally a
\emph{zero-copy} implementation.

\paragraph{Three examples:} Some examples of how the algorithm works,
in particular in what index order the temporary arrays are addressed
will surely be helpful.

Let $p=5\times 4 = 20=D[0]\times D[1]$ be a two-dimensional
factorization of $p$.  Algorithm~\ref{alg:dgenericalltoall} first
performs all-to-all operations on communication structures with
$D[0]=5$ processors, then on communication structures with $D[1]=4$
processors. The index sequences for the $S'$ and $R'$ arrays in each
round are shown in the table below. As can be seen, in each of the two
rounds, all $p$ indices $0,1,\ldots p-1$ are listed, meaning that $p$
blocks are sent and received for each processor $r=[O[0],O[1]]$,
except of course the blocks for processor $r$ itself. These are given
by $T[O[0]]$ for round $0$ and $T[O[1]]$ for round $1$. As can also be
seen, the blocks for the last round consist of consecutively indexed
elements.

\begin{center}
\begin{tabular}{lcl}
  & $k$ & \\
  \toprule
  & $0$ & $R'_{[1][5]}[5][4] = [0,1,2,3,4][0,5,10,15]$ \\
  \midrule
  $R'[0]$ & $=$ & $[0,5,10,15]$ \\
  $R'[1]$ & $=$ & $[1,6,11,16]$ \\
  $R'[2]$ & $=$ & $[2,7,12,17]$ \\
  $R'[3]$ & $=$ & $[3,8,13,18]$ \\
  $R'[4]$ & $=$ & $[4,9,14,19]$ \\
  \midrule
  & $1$ & $R'_{[5]}[4] = [0,5,10,15][0,1,2,3,4]$ \\
  \midrule
  $R'[0]$ & $=$ & $[0,1,2,3,4]$ \\
  $R'[1]$ & $=$ & $[5,6,7,8,9]$ \\
  $R'[2]$ & $=$ & $[10,11,12,13,14]$ \\
  $R'[3]$ & $=$ & $[15,16,17,28,19]$ \\
  \bottomrule
\end{tabular}
\end{center}

Now, let instead $p=2\times 3\times 4=24=D[0]\times D[1]\times D[2]$.
The index sequences for the $R'$ arrays in each round are shown in the
table below.  As can be seen, the blocks for the last round consist of
consecutively indexed elements.

\begin{center}
\begin{tabular}{lcl}
  & $k$ & \\
  \toprule
  & $0$ & $R'_{[1][2][6]}[2][3][4] = [0,1][0,2,4][0,6,12,18]$ \\
  \midrule
  $R'[0]$ & $=$ & $[0,6,12,18,2,8,14,20,4,10,16,22]$ \\
  $R'[1]$ & $=$ & $[1,7,13,19,3,9,15,21,5,11,17,23]$ \\
  \midrule
  & $1$ & $R'_{[2][6]}[3][4] = [0,2,4][0,6,12,18][0,1]$ \\
  \midrule
  $R'[0]$ & $=$ & $[0,1,6,7,12,13,18,19]$ \\
  $R'[1]$ & $=$ & $[2,3,8,9,14,15,20,21]$ \\
  $R'[2]$ & $=$ & $[4,5,10,11,16,17,22,23]$ \\
  \midrule
  & $2$ & $R'_{[6]}[4] = [0,6,12,18][0,1,2,3,4,5]$ \\
  \midrule
  $R'[0]$ & $=$ & $[0,1,2,3,4,5] $\\
  $R'[1]$ & $=$ & $[6,7,8,9,10,11] $\\
  $R'[2]$ & $=$ & $[12,23,14,15,16,17] $\\
  $R'[3]$ & $=$ & $[18,19,20,21,22,23] $\\
  \bottomrule
\end{tabular}
\end{center}

Finally, let $p=4\times 3\times 3\times 4=144=D[0]\times D[1]\times
D[2]\times D[3]$.  The index sequences for the $R'$ arrays in each
round are shown in the table below.  As can be seen, the blocks for
the last round consist of consecutively indexed elements.

\begin{center}
\begin{tabular}{lcl}
  & $k$ & \\
  \toprule
  & $0$ & $R'_{[1][4][12][36]}[4][3][3][4] = [0,1,2,3][0,4,8][0,12,24][0,36,72,108]$ \\
  \midrule
  $R'[0]$ & $=$ & $[0,36,72,108,12,\ldots,32,68,104,140] $\\
  $R'[1]$ & $=$ & $[1,37,73,109,13,\ldots,33,69,104,141] $\\
  $R'[2]$ & $=$ & $[2,38,74,110,14,\ldots,34,70,105,142] $\\
  $R'[3]$ & $=$ & $[3,39,75,111,15,\ldots,35,71,106,143] $\\
  \midrule
  & $1$ & $R'_{[4][12][36]}[3][3][4] = [0,4,8][0,12,24][0,36,72,108][0,1,2,3]$ \\
  \midrule
  $R'[0]$ & $=$ & $[0,1,2,3,36,37,38,39,\ldots,132,133,134,135] $\\
  $R'[1]$ & $=$ & $[4,5,6,7,40,41,42,43,\ldots,136,137,138,139] $\\
  $R'[2]$ & $=$ & $[8,9,10,11,44,45,46,47,\ldots,140,141,142,143] $\\
  \midrule
  & $2$ & $R'_{[12][36]}[3][4] = [0,12,24][0,36,72,108][0,1,2,3,4,5,6,7,8,9,10,11]$ \\
 \midrule
  $R'[0]$ & $=$ & $[0,1,2,3,4,5,6,7,8,9,10,11,36,\ldots,117,118,119] $\\
  $R'[1]$ & $=$ & $[12,13,14,15,16,17,18,19,20,21,22,23,48,\ldots,129,130,131] $\\
  $R'[2]$ & $=$ & $[24,25,26,27,28,29,30,31,32,33,34,35,60,\ldots,141,142,143] $\\
  \midrule
  & $3$ & $R'_{[36]}[4] = [0,36,72,108][0,1,2,3,4,5,6,7,8,\ldots,33,34,35]$ \\
  \midrule
  $R'[0]$ & $=$ & $[0,1,2,3,4,5,6,7,8,\ldots,33,34,35]$\\
  $R'[1]$ & $=$ & $[36,37,38,39,40,41,42,43,44,\ldots,69,70,71] $\\
  $R'[2]$ & $=$ & $[72,73,74,75,76,77,78,79,80,\ldots,105,106,107] $\\
  $R'[3]$ & $=$ & $[108,109,110,111,112,113,114,115,116,\ldots,141,142,143] $\\
  \bottomrule
\end{tabular}
\end{center}

\section{The Implementation in MPI}

For the implementation of Algorithm~\ref{alg:dgenericalltoall} in and
for MPI, three details have to be handled:
\begin{enumerate}
\item
  Creation of the communication structures $G$.
\item
  Traversal of the arrays of blocks of elements.
\item
  The component all-to-all communications.
\end{enumerate}
We give the concrete, full MPI code in C which handles the three
issues, showing how the non-trivial all-to-all algorithm can be
effectively implemented with MPI.

\begin{lstlisting}[caption={Full MPI code for (graph) factorizing a
      communicator by communicator splitting following a factorization
      of the number of processes given as an array of dimension
      orders.},label=lst:mpifactor]
int Comm_factorize(MPI_Comm comm,
                   const int d, const int order[],
                   MPI_Comm *factorcomm)
{
  int rank, size;
  int i, p;

  torusattr *factors;

  p = 1;
  for (i=0; i<d; i++) {
    assert(order[i]>1);
    p *= order[i];
  }

  int periodicity[d];
  for (i=0; i<d; i++) periodicity[i] = 1;
  // make Cartesian communicator with no reorder
  MPI_Cart_create(comm,d,order,periodicity,0,factorcomm); 
  int result;
  MPI_Comm_compare(comm,*factorcomm,&result);
  assert(result==MPI_CONGRUENT);
  
  MPI_Comm_rank(*factorcomm,&rank);
  MPI_Comm_size(*factorcomm,&size);
  assert(size==p);
  
  factors = (torusattr*)malloc(sizeof(torusattr));
  assert(factors!=NULL);
  factors->d = d;
  factors->comm = (MPI_Comm*)malloc(d*sizeof(MPI_Comm));
  assert(factors->comm!=NULL);

  int origin[d];
  MPI_Cart_coords(*factorcomm,rank,d,origin);
  
  for (i=0; i<d; i++) {
    int f;
    p /= order[i]; f = rank-origin[i]*p; 
    MPI_Comm_split(*factorcomm,f,origin[i],&factors->comm[i]);
    int factorsize;
    MPI_Comm_size(factors->comm[i],&factorsize);
    assert(factorsize==order[i]);
  }  
  MPI_Comm_set_attr(*factorcomm,toruskey(),factors);

  return MPI_SUCCESS;
}
\end{lstlisting}

Given a communicator of size $p$ and a factorization of $p$ into $d$
factors, the code in Listing~\ref{lst:mpifactor} shows how to create
subcommunicators corresponding to the subgraphs $G$ of
Algorithm~\ref{alg:dgenericalltoall} to be used by the processes in
the later implementation of the all-to-all algorithm.  For each MPI
process, for all-to-all communication round $k, k=0,1,\ldots,d-1$, a
communicator consisting of the process itself and all other processes
having the same coordinates in the torus except for the $k$th
coordinate (which is between $0$ and $\texttt{order}[k]$) is created
by \mpicommsplit.  Processes that belong together are colored by the
first process which has a $k$th coordinate of zero, and the call will
compute all $p/\mathtt{order}[k]$ communicators (subgraphs) for round
$k$. To associate $d$-dimensional coordinates with the processes, a
Cartesian MPI communicator is created for convenience.

The creation of the communicators which entails $d$ \mpicommsplit
calls is expensive. For each process, the $d$ communicators to which
this process belongs are therefore cached with the calling communicator
for later lookup when the all-to-all operation is performed. For
caching, a communicator key value must be created; a possible way of
doing this transparently is shown in Listing~\ref{lst:mpicache}.

\begin{lstlisting}[caption={Surrounding code for caching subcommunicators
      with a communicator. A key value is kept as a static lookup
      value that is created at first use. A delete function (beware: with
      collective semantics) frees created subcommunicators for
      the process.},label=lst:mpicache]
typedef struct {
  int d;
  MPI_Comm *comm;
} torusattr;

static int torusdel(MPI_Comm comm,
                    int keyval, void *attr, void *s)
{
  torusattr *factors = (torusattr*)attr;

  int i;
  for (i=0; i<factors->d; i++)
    MPI_Comm_free(&factors->comm[i]);
  free(factors->comm);
  free(factors);
  
  return MPI_SUCCESS;
}

static int toruskey() {
  static int toruskeyval = MPI_KEYVAL_INVALID; // hidden key

  if (toruskeyval==MPI_KEYVAL_INVALID) {
    MPI_Comm_create_keyval(MPI_COMM_NULL_COPY_FN,
                           torusdel,&toruskeyval,NULL);
  }
  return toruskeyval;
}
\end{lstlisting}

\begin{lstlisting}[caption={MPI implementation of the $d$ factor all-to-all
      algorithm. A datatype for the sequence of blocks to be communicated is
      created for each communication round,
      and freed again.},label=lst:mpialltoall]
int Alltoall_torus(void *sendbuf,
                   int sendcount, MPI_Datatype sendtype,
                   void *recvbuf,
                   int recvcount, MPI_Datatype recvtype,
                   MPI_Comm comm)
{
  int rank, size;
  int i, k;
  int c, d;

  MPI_Comm_rank(comm,&rank);
  MPI_Comm_size(comm,&size);
  
  MPI_Datatype block;
  MPI_Aint lb, blockextent;
  MPI_Type_contiguous(recvcount,recvtype,&block);
  MPI_Type_get_extent(block,&lb,&blockextent);

  void *tempbuf;
  tempbuf = (void*)malloc(size*blockextent);
  assert(tempbuf!=NULL);
  
  int flag;
  torusattr *factors;
  MPI_Comm_get_attr(comm,toruskey(),&factors,&flag);
  if (!flag) // fallback
    return MPI_Alltoall(sendbuf,sendcount,sendtype,
                        recvbuf,recvcount,recvtype,comm);
  
  int count, stride[factors->d];

  count = 1;
  for (i=factors->d-1; i>=0; i--) {
    stride[i] = count;
    MPI_Comm_size(factors->comm[i],&c);
    count = c*count;
  }

  void *in, *out;
  out = sendbuf;
  in = (factors->d%2==0) ? tempbuf : recvbuf;
  
  count = 1;
  for (k=factors->d-1; k>=0; k--) {
    // create the contiguous segment
    MPI_Datatype segment, tile, tiled;
    MPI_Type_contiguous(count,block,&segment);

    // column-major parts
    for (i=0; i<k; i++) {
      MPI_Type_create_resized(segment,0,stride[i]*blockextent,
                             &tile);
      MPI_Type_free(&segment);
      MPI_Comm_size(factors->comm[i],&c);
      MPI_Type_contiguous(c,tile,&segment);
      MPI_Type_free(&tile);
    }
    MPI_Type_create_resized(segment,0,count*blockextent,
                            &tiled);
    MPI_Type_free(&segment);
    MPI_Type_commit(&tiled);

    MPI_Alltoall(out,1,tiled,in,1,tiled,factors->comm[k]);
    MPI_Type_free(&tiled);

    if (out==sendbuf) {
      if (in==recvbuf) {
        out = recvbuf; in = tempbuf;
      } else {
        out = tempbuf; in = recvbuf;
      }
    } else {
      void *s;
      s = in; in = out; out = s;
    }

    MPI_Comm_size(factors->comm[k],&c);
    count = c*count;
  }

  MPI_Type_free(&block);
  free(tempbuf);

  return MPI_SUCCESS;
}
\end{lstlisting}

Finally, the implementation of the all-to-all algorithm in
Algorithm~\ref{alg:dgenericalltoall} is shown as
Listing~\ref{lst:mpialltoall}. The factorization with communicators to
be used by the calling process is looked up using the attribute key.
We here assume that the communicator has indeed been factorized, if not
\mpialltoall on the full communicator is performed as fallback.

The implementation goes through $d$ iterations in each of which an
\mpialltoall call is done. The send and receive buffers \texttt{in}
and \texttt{out} are traversed in the order described by the stride
function $\sigma(i)$.  These strides are precomputed into the
$\mathtt{stride}[d]$ array, and used to create the datatype for each
iteration. First, a single block of elements is described by the
\texttt{block} type. Input and output buffers are such blocks in
rank order, but the order in which these blocks are sent and received
is determined by the final \texttt{tiled} datatype.  Strides are set
by resizing component datatypes with \mpiresized, these strided
segments of blocks are listed in contiguous order by creating
contiguous, but tiled types with \mpicontig. As can be seen, the
resulting implementation is completely zero-copy in the sense that there
are no process-local explicit copying of data whatsoever. This is
accomplished by taking blocks directly from the \texttt{sendbuf} in the
first iteration and making sure that \texttt{recvbuf} is used as
receive buffer in the last iteration. A theoretical drawback of the
implementation is that $O(d)$ time is spent creating datatypes in each
iteration before the \mpialltoall call. With the static, derived
datatype mechanism of MPI, this cannot be avoided, but in a more
flexible programming model, it possibly could.

\section{Selected Benchmark Results}

We have implemented the generic all-to-all algorithm exactly as
explained and shown in the previous sections. We aim to compare it
against itself with different factorizations of the number of
processes $p$ into $d$ factors, and of course against the MPI library
native \mpialltoall operation for the same inputs. For the latter, we
examine three MPI libraries that are available to us.

Our test system is a small $36\times 32$ processor cluster with
$36$~dual socket compute nodes, each with two Intel(R) Xeon(R) Gold
6130F $16$-core CPUs. The nodes are interconnected via dual Intel
OmniPath interconnects each with a bandwidth of $100$ GBytes/s. The
implementations and benchmarks were compiled with \gcc with the
\texttt{-O3} option.  The MPI libraries used are \hydraopenmpi,
\hydrampich, and \hydraintelmpi.

Our benchmark does a simple \mpialltoall call with preallocated and
filled send and receive communication buffers in deciles that are
increased by factors of $10$ from $1$ to $10\, 000$ elements per MPI
process. Elements are integers of MPI datatype \mpiint. For each
number of elements, $40$ measurement repetitions are done with $8$
prior warmup measurements. The best seen completion times (of the
slowest process) are plotted. Processes are synchronized before each
measurement with a double \mpibarrier call.

We experiment with different number of dimensions from $2$ to
$\ceil{\log_2 p}$ and use \mpidimscreate to do the actual
factorization of the number of processes. We assume that
\mpidimscreate behaves according to the MPI
specification~\cite{Traff15:dimscreate} and returns a factorization
where the factors are as close to each other as possible. This is not
quite the case for the \hydraopenmpi library as can be seen in
Table~\ref{tab:factorization}. This problem was addressed
in~\cite{Traff15:dimscreate} and can be solved. 

\begin{table}
  \caption{The factorizations of the given number of processes $p$ for
    different number of dimensions $d$ as used in the experiments. The
    factorization is computed (presumably) in a portable way by
    \mpidimscreate. As can be seen, the \hydraopenmpi library computes
    a factorization that is not in accordance with the MPI
    specification.}
  \label{tab:factorization}
\begin{center}
\begin{tabular}{rcc}
  $d$ & \hydraintelmpi & \hydraopenmpi \\
  & \hydrampich & \\
  \toprule
  $2$ & $p=36\times 32$ & $p=48\times 24$\\
  $3$ & $p=12\times 12\times 8$ & $p=12\times 12\times 8$ \\
  $4$ & $p=8\times 6\times 6\times 4$ & $p=8\times 6\times 6\times 4$ \\
  $\ceil{\log_2 p}$ & $p=3\times 3\times 2\times 2\times 2\times 2\times 2\times 2\times 2$ &
  $p=3\times 3\times 2\times 2\times 2\times 2\times 2\times 2\times 2$ \\
  \bottomrule
\end{tabular}
\end{center}
\end{table}

\begin{figure}
  \begin{center}
    \includegraphics[width=\linewidth]{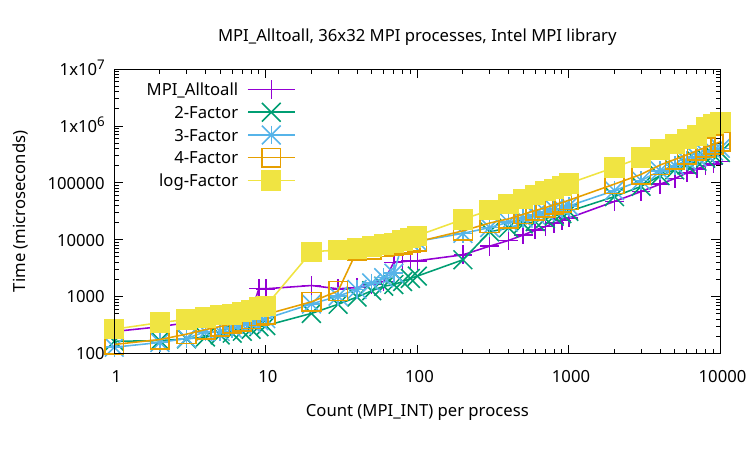}
  \end{center}
  \caption{Results for the \hydraintelmpi library.}
  \label{fig:alltoallintel}
\end{figure}

\begin{figure}
  \begin{center}
    \includegraphics[width=\linewidth]{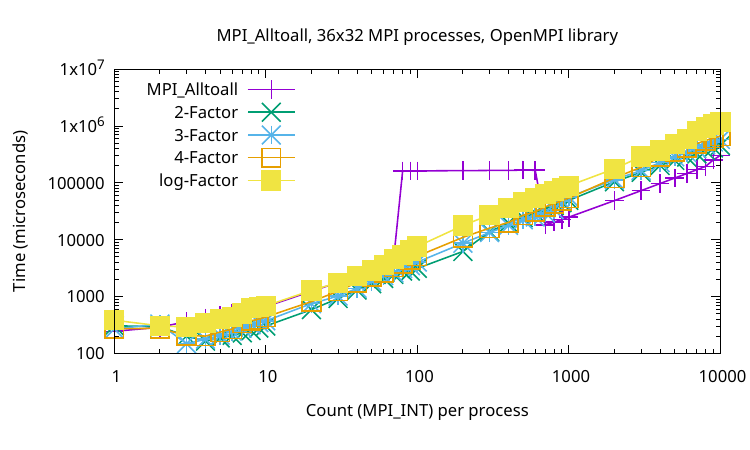}
  \end{center}
  \caption{Results for the \hydraopenmpi library.}
  \label{fig:alltoallopenmpi}
\end{figure}

\begin{figure}
  \begin{center}
    \includegraphics[width=\linewidth]{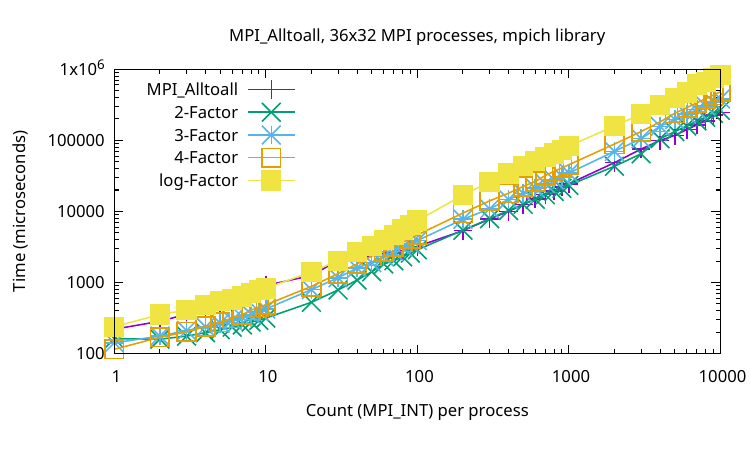}
  \end{center}
  \caption{Results for the \hydrampich library.}
  \label{fig:alltoallmpich}
\end{figure}

We run the new all-to-all implementation in a single full-system MPI
process configuration with $p=36\times 32$ and compare different
factorizations of $p$ against the library native \mpialltoall.  We
have tried with $d=2,3,4$ factors and $\ceil{\log_2 p}=9$ factors, see
Table~\ref{tab:factorization} for the specific factorizations.
Results with \hydraintelmpi are shown in
Figure~\ref{fig:alltoallintel}, with \hydraopenmpi in
Figure~\ref{fig:alltoallopenmpi}, and with \hydrampich in
Figure~\ref{fig:alltoallmpich}. The different MPI libraries lead to
similar conclusions, but have many specific and conspicuous
differences. First, with up to, say, $100$ \mpiint elements per
process, the $d=2,3$ factor algorithm seems to be considerably faster
than library native \mpialltoall, in many cases by a factor of two or
more. For larger number of elements, a (presumably) direct algorithm
is faster. The highest possible dimensional $d=\ceil{\log_2 p}$ round
algorithm is never competitive. The \hydraopenmpi library exhibits a
severe \mpialltoall performance (guideline violation) problem for
medium sized blocks from about $80$ to $800$ \mpiint elements, where
performance unexpectedly degrades by a factor of more than
$10$. Therefore, MPI libraries may have room for improvement by
message-combining implementations for smaller inputs.

\section{Summary}

We presented an extremely compact, highly non-trivial algorithm for
all-to-all communication using a factorization of the number of
processes to decompose the problem into smaller component all-to-all
problems. We showed how this algorithm could be effectively and
transparently be implemented in MPI by using derived datatypes,
subcommunicators and attribute caching on communicators. The resulting
algorithm could be used as performance guideline against which to
compare library native \mpialltoall implementations, and as a template
for possibly more efficient implementations in MPI libraries by
careful selection of the algorithms for the component \mpialltoall
calls and by choosing a factorization and splitting of the
communicators that map well to the given, physical communication
system.

\bibliographystyle{plain}
\bibliography{parallel,traff}

\end{document}